\documentclass[a4paper,12pt]{article}

\usepackage{amsmath, wrapfig}
\usepackage{dsfont, a4wide, amsthm, amssymb, amsfonts, graphicx}
\usepackage{fancyhdr, xspace, psfrag, setspace, supertabular, color, enumerate}
\usepackage{hyperref}

\newtheorem{theorem}{Theorem}[section]
\newtheorem{proposition}[theorem]{Proposition}
\newtheorem{lemma}[theorem]{Lemma}

\newtheorem{corollary}[theorem]{Corollary}

\newtheorem{definition}[theorem]{Definition}

\newtheorem{example}[theorem]{Example}

\def\a{\alpha}

\def\g{\gamma}

\def\e{\epsilon}
\def\t{\theta}

\def\F{\mathbb{F}}

\def\P{\mathbb{P}}

\def\H{\mathbb{H}}

\title{On entropy Marton-type inequalities and small symmetric differences with cosets of abelian groups}
\author{Thomas Karam\footnote{Mathematical Institute, University of Oxford. Email: \texttt{thomas.karam@maths.ox.ac.uk}.}}

\begin{document}
\maketitle

\begin{abstract}

We recognise that an entropy inequality akin to the main intermediate goal of recent works (Gowers, Green, Manners, Tao \cite{GGMT1}, \cite{GGMT2}) regarding a conjecture of Marton provides a black box from which we can also through a short deduction recover another description: if a finite subset $A$ of an abelian group $G$ is such that the distribution of the sums $a+b$ with $(a,b) \in A \times A$ is only slightly more spread out than the uniform distribution on $A$, then $A$ has small symmetric difference with some finite coset of $G$. The resulting bounds are necessarily sharp up to a logarithmic factor.

\end{abstract}

\tableofcontents

\section{Distributional small doubling}

A central theme in contemporary combinatorics is that of results of the following kind. Given an abelian group $(G,+)$ and a finite subset $A$ of $G$ such that $A+A$ is not too large compared to $A$, what can be said about the structure of $A$ ?

In the integers, it is well-known that the smallest possible size of $A+A$ is equal to $2|A|-1$, which is attained when $A$ is an arithmetic progression. Conversely, Freiman’s theorem states that if for some $K \ge 1$ we have $|A+A| \le K|A|$, then $A$ is contained in an arithmetic progression with dimension at most $C_1(K)$ and length at most $C_2(K)$ for some integers $C_1(K)$,$C_2(K)$.

Still intuitively more convenient perhaps than the case of the integers is the case of finite abelian groups, where the smallest possible value of $|A+A|$ is equal to $|A|$, with equality being attained if and only if $A$ is a coset of $G$. Analogues of Freiman’s theorem have been studied in various settings, but even the simplest cases, such as that of $\F_2^n$, are interesting in their own right: as discussed by Sanders \cite{Sanders F_2^n} this case has applications to coding theory, and very recently this case was used in a breakthrough \cite{GGMT1} of Gowers, Green, Manners and Tao as a first case in which they were able to prove in a first form a conjecture of Marton also known as the polynomial Freiman-Ruzsa conjecture, using ideas that are simpler to explain (and hence more amenable to discovery as well) before moving to arguments building on these ideas in a more complicated way in the more general case of groups with bounded torsion \cite{GGMT2}. The authors of \cite{GGMT1}, \cite{GGMT2} proved in particular the following.

\begin{theorem} [\cite{GGMT2}, part of Theorem 1.1] \label{Marton’s conjecture} Let $m \ge 2$ be an integer, let $K \ge 1$, and let $G$ be an abelian group such that $mx=0$ for every $x \in G$. If $A$ is a non-empty subset of $G$ such that $|A+A| \le K|A|$, then there exists a subgroup $H$ of $G$ with size at most $|A|$ such that $A$ is contained in a union of at most $(2K)^{O(m^3)}$ cosets of $H$. \end{theorem}

In a different direction, the ratio $|A+A|/|A|$ in the general setting of non-necessarily abelian groups had ultimately led to structure results in the theory of approximate groups by Breuillard, Green and Tao \cite{Breuillard Green Tao}, where a core result, Theorem 1.6, in part says informally that if $|A+A| \le K|A|$ then $A$ can be covered by at most some $C(K)$ left-translates of a finite-by-nilpotent subgroup of $G$.

The reader may have noticed by now that in all of the previously discussed statements, the conclusion on the description of $A$ always provides a structure inside which $A$ is contained, rather than a structure which has a small symmetric difference with $A$. This is not surprising if, given a finite subgroup $H$ with a fixed and substantial size, we consider the effect of taking small perturbations of $H$ that either add to or remove from $H$ a very small number of elements, as there is a major qualitative difference between how fast the ratio $|A+A|/|A|$ moves away from $1$ in both cases. 

\begin{example} \label{Asymmetry for almost subgroups from the inside and the outside} If $A$ is contained in $H$, then that guarantees that $A+A$ is contained in $H+H=H$, and the ratio $|A+A|/|A|$ is hence at most \[|H|/|A| = 1 + |A \Delta H|/|H| + o(|A \Delta H|/|H|) \] when $|A \Delta H|$ is small compared to $|H|$. On the other hand if $A$ contains $H$ and we have the disjoint union \[A = H \cup \{x_1, \dots, x_k\} \] for some $x_1, \dots, x_k \in G \setminus H$ that furthermore each belong to different cosets of $H$, then $A+A$ contains in particular the disjoint union of cosets \[H \cup (H+\{x_1\}) \cup \dots \cup (H+\{x_k\})\] which has size $|H| (1+|A \Delta H|)$. \end{example}

More generally, if $|A \Delta H|$ is small compared to $|H|$ then the size of $A+A$ depends much more strongly on the elements of $A \setminus H$ than on the elements of $H \setminus A$. This asymmetry, between finite subsets $A$ of a group $G$ that are in a sense almost all of a finite subgroup from the inside and almost all of a finite subgroup from the outside, disappears to a large extent if instead of $|A+A|$ we consider the distribution of $a+b$ with $a,b$ chosen independently at random according to the uniform distribution on $A$. Indeed, if we take $A$ to be $H \cup \{x_1, \dots, x_k\}$ for some finite subgroup $H$ of $G$ and some $x_1, \dots, x_k \in G \setminus H$, then only a proportion at most \[\frac{2k|H|+k^2}{(|H|+k)^2} \le \frac{2k}{|H|+k}\] of the pairs $(a,b) \in A^2$ are outside $H^2$, and hence so is the proportion of pairs $(a,b) \in A^2$ such that $a+b \notin H$.

We may hence aspire to prove results of the following type: if $A$ is a finite subset of an abelian group $G$ such that the distribution of $a+b$ with $(a,b) \in A^2$ chosen uniformly at random is in some sense close to the uniform distribution on $A$, then there exists a finite subgroup $H$ of $G$ such that the symmetric difference $A \Delta H$ has a size which is small compared to that of $A$.

A reformulation of the inequality $|A+A| \le K|A|$ is that there exists a subset $U$ of $G$ with size $|A|$ such that \[|(A+A) \setminus U| \le (K-1)|A|,\] and this reformulation immediately suggests an analogous distributional definition.

\begin{definition}\label{e-approximate group}

Let $\e>0$, let $G$ be an abelian group, and let $A$ be a non-empty finite subset of $G$. We say that $A$ is a \textit{distributional $\e$-approximate group} of $G$ if there exists a subset $U$ of $G$ with size $|A|$ such that the proportion of pairs $(a,b) \in A \times A$ with $a+b \in U$ is at least $1-\e$.

\end{definition}

Contrary to what happens for the sizes of sumsets, the effect of adding or deleting a few elements of a finite subgroup of substantial size on the smallest $\e$ such that the resulting set is a distributional $\e$-approximate group is qualitatively similar.

\begin{example} \label{Estimate on symmetric differences}
Let $G$ be an abelian group, let $H$ be a finite subgroup of $G$, and let $T \subset G$. As $|T|/|H|$ tends to $0$ the following holds. \begin{enumerate}[(i)]
\item If $T \subset H$ then the smallest $\e$ such that $H \setminus T$ is a distributional $\e$-approximate group is \[|T|/|H|+o(|T|/|H|).\]
\item If $T$ is disjoint from $H$ then the smallest $\e$ such that $H \cup T$ is a distributional $\e$-approximate group is \[2|T|/|H|+o(|T|/|H|).\]
\end{enumerate}
\end{example}

\begin{proof} Item (i) follows from the fact that for every $z \in H$ the number of pairs $(x,y) \in (H \setminus T)^2$ with $x+y=z$ is between $|H|-2|T|$ and $|H|$, and item (ii) follows from the fact that for every $z \in G \setminus H$ the number of pairs $(x,y) \in (H \cup T)^2$ with $x+y=z$ is at most $2|T|$. \end{proof}

In general, we will show that if a finite subset $A$ of an abelian group $G$ is a distributional $\e$-approximate group for some small $\e$, then $A$ has small symmetric difference with some finite coset of $G$. More formally, we have the following bound, which Example \ref{Estimate on symmetric differences} shows is sharp in the dependence on $\e$ up to the log factor. (We have not attempted to optimise the constant 240.)

\begin{theorem} \label{Main theorem}

There exists $\t>0$ such that if $\e \in (0,\exp(-2))$, $G$ is an abelian group, and $A \subset G$ is a non-empty finite distributional $\e$-approximate group, then \[\frac{|A \Delta (H + \{x\})|}{|H|} \le (240 \log(\e^{-1}|A|)) \e\] for some finite coset $H +\{x\}$ of $G$, provided that the right-hand side is at most $\t$.

\end{theorem}

Variants of Theorem \ref{Main theorem} have attracted considerable interest. To give one example, an analogue of Theorem \ref{Main theorem} where the assumption pertains to the additive energy, that is, the number of quadruples $(a,b,a',b') \in A^4$ satisfying $a+b=a’+b’$, is the starting point (Proposition 1.1) of a 2013 ECM survey \cite{Sanders} of Sanders. We have chosen to use Definition \ref{e-approximate group} instead of additive energy as it appears to us natural as a closest modification of the small doubling assumption $|A+A| \le K|A|$ which does not present the qualitative asymmetry illustrated by Example \ref{Asymmetry for almost subgroups from the inside and the outside}, but the proof structure used in this paper can also be easily adapted to obtain an analogue of Lemma \ref{bound on Delta H} and then of Theorem \ref{Main theorem} where the assumption is that the additive energy of $A$ is at least $(1-\e) |A|^3$ (with adapted bounds in the conclusion).

Our main aim will be to show how Theorem \ref{Main theorem} can be obtained by using as a black box a variant of a recent entropic result which was used as the core intermediate step in the proof of Theorem \ref{Marton’s conjecture} .

In Section \ref{Section: Shannon entropy and information-theoretic inequalities} we will discuss some information-theoretic background that allows us to state and use this black box. Then, in Section \ref{Section: Deduction of the main result from the entropic Marton conjecture} we will deduce Theorem \ref{Main theorem} from it.

\section{Information-theoretic and Marton-type inequalities} \label{Section: Shannon entropy and information-theoretic inequalities}

We begin this section by discussing a few information-theoretic facts. Throughout, we will say that the \emph{range} of a random variable $X$ is the set of $x$ such that $\P(X=x) > 0$. If $X$ is a random variable with finite range $A$, then the Shannon entropy of $X$ is defined as \[ \H(X) = - \sum_{x \in A} \P(X=x) \log \P(X=x). \] It is well-known that we have the upper bound \begin{equation} \H(X) \le \log |A|. \label{bound on entropy in terms of range} \end{equation} Furthermore, when $X$ is uniformly distributed on $A$ this inequality becomes an equality, which provides a connection between the entropy of the variable $X$ and the size of the set $A$. If $X,Y$ are random variables, then it is always the case that \begin{equation} \H((X,Y)) \le \H(X) + \H(Y) \label{subadditivity of entropy} \end{equation} with equality if and only if $X,Y$ are independent. If $X,Y$ are random variables then we define the conditional entropy \[\H(X|Y) = \H(X,Y) - \H(Y),\] which by \eqref{subadditivity of entropy} always satisfies \begin{equation} \H(X|Y) \le \H(X). \label{conditioning does not increase entropy} \end{equation} 

If $A$ is a finite set, we will throughout write $U_A$ for the random variable with probability mass $1/|A|$ on every element of $A$ and probability mass $0$ everywhere else.

The probability distribution of $a+b$ with $(a,b) \in A^2$ chosen uniformly at random may now be written as the distribution of $U_A + U_A$, with the two copies of $U_A$ implicitly taken to be independent, as we will do throughout the entire paper without recalling that.

One way of measuring how much the distribution of $U_A$ spreads out is the difference of the entropies \[\Delta \H(A):= \H(U_A + U_A) - \H(U_A).\]

The difference $\Delta \H(A)$ is always nonnegative, although that might not be immediately obvious. After all, as discussed previously it is well-known that among random variables taking values in some finite set $D$, the random variable $U_D$ \textit{maximises} the entropy, so for $\H(U_A+U_A)$ to be at least $\H(U_A)$, the fact that the range of $U_A+U_A$ is larger compared to that of $U_A$ must at least compensate the fact that $U_A+U_A$ is no longer necessarily uniformly distributed on its range. Nonetheless, the desired nonnegativity does hold, as an interpretation in terms of conditional entropy will show in the proof of the next lemma (for which it goes without saying that we do not claim any originality whatsoever). We note that the case where the ranges of $U_A+U_A$ and $U_A$ have the same size corresponds to the case where $A$ is a finite coset of $G$, and $U_A+U_A$ is hence then perfectly uniform on its range $A+A$, so we have $\H(U_A+U_A) = \H(U_A)$.

\begin{lemma} \label{independent addition increases entropy} Let $X$,$Y$ be two independent random variables with finite range. Then we have the lower bound \[\H(X+Y) \ge \max(\H(X),\H(Y)).\] \end{lemma}

\begin{proof} We interpret $\H(X)$ as the conditional entropy $\H(X+Y|Y)$. It follows from \eqref{conditioning does not increase entropy} that \[\H(X+Y|Y) \le \H(X+Y),\] so we conclude that $\H(X)$ (and likewise $\H(Y)$) is at most $\H(X+Y)$. \end{proof}

If $X$,$Y$ are random variables with finite range taking values in an abelian group, then the \emph{entropic Rusza distance}, defined by Rusza \cite{Ruzsa}, then studied in more detail by Tao \cite{Tao} and recently again by Green, Manners and Tao \cite{GMT} is the quantity \[d(X,Y) = \H(X’+Y’) - (\H(X’)+\H(Y’))/2\] where $X’,Y’$ are independent and each distributed as $X,Y$ respectively. Lemma \ref{independent addition increases entropy} shows that $d(X,Y)$ is always nonnegative.

In general we do not have $d(X,X) = 0$, and it is also false, strictly speaking, that if $d(X,Y)$ then $X=Y$: indeed, if $X$,$Y$ are uniformly distributed on any two different cosets $H + \{x_1\}, H + \{x_2\}$ of the same finite subgroup $H$ of $G$, then $X+Y$ is uniformly distributed on $H + \{x_1+x_2\}$ and the three entropies $\H(X+Y), \H(X), \H(Y)$ are hence all equal to $\log |H|$, so we have $d(X,Y) = 0$. This, nonetheless, will not get in our way.

In \cite{GGMT1} and then \cite{GGMT2}, Gowers, Green, Manners and Tao established the following entropic version of Theorem \ref{Marton’s conjecture}.

\begin{proposition} [\cite{GGMT2}, Theorem 1.3] \label{Entropic Marton conjecture} There exists an absolute constant $C>0$ such that the following holds. If $m \ge 2$ is a positive integer, $G$ is an abelian group satisfying $mx=0$ for all $x \in G$, and $X,Y$ are random variables with finite range taking values in $G$ then there exists a finite subgroup $H$ of $G$ such that \[d(X, U_H) \le Cm^3 d(X,Y) \text{ and } d(Y, U_H) \le Cm^3 d(X,Y).\] \end{proposition}

The result that we will use to prepare our black box will be the following slight variant of Proposition \ref{Entropic Marton conjecture}, proved by Green, Manners and Tao \cite{GMT} a few months before \cite{GGMT1} and \cite{GGMT2}. Compared to Proposition \ref{Entropic Marton conjecture}, this variant has the explicit limitation that it only pertains to $\e$ sufficiently small, but as the bound in the conclusion of Proposition \ref{Entropic Marton conjecture}, which we will use only when it is at most $1/10$, involves an multiplicative constant $C$, Proposition \ref{Entropic Marton conjecture} presents for our purposes a very similar limitation. Furthermore, Proposition \ref{Entropic Marton conjecture for small values} pertains to all abelian groups and its bounds are uniform over all abelian groups, so will for us provide a stronger application than Proposition \ref{Entropic Marton conjecture} would. (Finally, it may be noted that the proof of Proposition \ref{Entropic Marton conjecture for small values} is not as involved as that of Proposition \ref{Entropic Marton conjecture}.)

\begin{proposition} [\cite{GMT}, Proposition 1.3] \label{Entropic Marton conjecture for small values} There exists an absolute constant $\g>0$ such that the following holds. If $G$ is an abelian group and $X,Y$ are random variables with finite range taking values in $G$ and satisfying $d(X,Y) \le \g$ then there exists a finite subgroup $H$ of $G$ such that \[d(X, U_H) \le 12 d(X,Y) \text{ and } d(Y, U_H) \le 12 d(X,Y).\] \end{proposition}

By taking $X=Y$ to be of the type $U_A$ for some non-empty finite subset $A$ of $G$, we obtain a connection between $\Delta \H(A)$ and $d(U_A,U_H)$ for some finite subgroup $H$ of $G$, which is the black box that we will use to obtain our main result, Theorem \ref{Main theorem}.

\begin{corollary} \label{Corollary with X=Y}
There exists an absolute constant $\g>0$ such that the following holds. If $G$ is an abelian group and $A$ is a non-empty finite subset of $G$ satisfying $\Delta \H(A) \le \g$, then \[d(U_A, U_H) \le 12 \Delta \H(A)\] for some finite subgroup $H$ of $G$.
\end{corollary}

\section{Obtaining a small symmetric difference with a coset} \label{Section: Deduction of the main result from the entropic Marton conjecture}

In this section we prove Theorem \ref{Main theorem}. The proof will involve three stages. We will first show that if $\e$ is small and $A \subset G$ is a (non-empty, finite) distributional $\e$-approximate group then $\Delta \H(A)$ is small. The black box, Corollary \ref{Corollary with X=Y}, then allows us to deduce that $d(U_A, U_H)$ is small for some finite subgroup $H$ of $G$. As we discussed previously, having $d(U_A, U_H)$ small (or even zero) does not guarantee that $A$ has small symmetric difference with $H$. But, as we will show in the third stage, Lemma \ref{bound on A Delta H}, this is necessarily the case with some coset of $H$.

\begin{lemma} \label{bound on Delta H} Let $\e \in (0,\exp(-2))$, let $G$ be an abelian group, and let $A \subset G$ be a non-empty distributional $\e$-approximate group. Then \[\Delta \H(A) \le 2\e \log (\e^{-1}|A|).\] \end{lemma}

\begin{proof}

Let $X$ be the random sum $X=a+b$ with $(a,b) \in A^2$ chosen uniformly at random. The range of $X$ has size at most $|A|^2$, so we can partition it into two disjoint subsets $U$,$V$ of $G$ satisfying the four conditions \[|U| = |A|\text{, } |V| \le |A|^2\text{, }\P(X \in U) \ge 1-\e \text{, }\P(X \in U \cup V) = 1.\] We decompose \begin{align*}\H(X) = & - \sum_{x \in U} \P(X=x) \log \P(X=x) - \sum_{x \in V} \P(X=x) \log \P(X=x)\\
= & - \P(X \in U) \log \P(X \in U) - \P(X \in U) \sum_{x \in U} \P(X=x|X \in U) \log \P(X=x|X \in U)\\
& - \P(X \in V) \log \P(X \in V) - \P(X \in V) \sum_{x \in V} \P(X=x|X \in V) \log \P(X=x|X \in V)\\
= & \P(X \in U) \H(X|X \in U) + \P(X \in V) \H(X|X \in V) - (p \log p + (1-p) \log (1-p)) \end{align*}
where $p = \P(X \in V) \le \e$. By the upper bounds \[\H(X|X \in U) \le \log |U|, \H(X|X \in V) \le \log |V| \le 2 \log |A|\] coming from \eqref{bound on entropy in terms of range}, as well as
\[(p \log p + (1-p) \log (1-p)) \le - \e \log \e - 2\e \le - 2\e \log \e\]
where the last inequality uses $\e \le \exp(-2)$, we deduce \[ \H(X) \le \log |U| + \e \log |V| - 2\e \log \e.\] Since $|U|=|A|$ and $|V| \le |A|^2$, this concludes the proof. \end{proof}

The next lemma is what will allow us to move from the upper bound on the (entropic) Rusza distance in the conclusion of Corollary \ref{Corollary with X=Y} to the upper bound on the size of the symmetric difference in the conclusion of Theorem \ref{Main theorem}.

\begin{lemma} \label{bound on A Delta H} Let $\a \in [0, 1/10)$, let $G$ be an abelian group, let $A$ be a non-empty finite subset of $G$, and let $H$ be a finite subgroup of $G$. If $d(U_A, U_H) \le \a$ then \[|A \Delta (H + \{x\})| \le 10 \a |H|\] for some $x \in G$.\end{lemma}

\begin{proof}

For every $x \in G$ we have the identities \begin{align*} |H| - |A| &= |(H + \{x\}) \setminus A| - |A \setminus (H + \{x\})|\\
|A \Delta (H + \{x\})| &= |A \setminus (H + \{x\})| + |(H + \{x\}) \setminus A| \end{align*} and therefore the decomposition \begin{equation} |A \Delta (H + \{x\})| = (|H| - |A|) + 2|A \setminus (H + \{x\})|.\label{two differences} \end{equation} We shall upper bound both parts of the right-hand side of \eqref{two differences} separately for some $x \in G$.

By definition we can rewrite \[d(U_A,U_H) = \frac{\H(U_A+U_H) - \H(U_A)}{2} + \frac{\H(U_A+U_H) - \H(U_H)}{2}.\] By Lemma \ref{independent addition increases entropy} both contributions are nonnegative, so it follows that \begin{equation} \H(U_A+U_H) - \H(U_H) \le 2 \a \label{first inequality} \end{equation} but also that \[ |\H(U_A) - \H(U_H)| \le 2 \a.\] The second inequality simplifies to \[ | \log |A| - \log |H|| \le 2\a\] and hence provides \begin{equation} |H|-|A| \le |H| (1-\exp(-2\a)) \le 2 \a |H|. \label{size difference} \end{equation}

We now move to upper bounding $|A \setminus (H + \{x\})|$ for some suitable $x \in G$. We may without loss of generality identify $G$ with $H \times (G/H)$, so that we may write every $g \in G$ as $g = (h(g),b(g))$ with $h \in H$, $b \in G/H$. We write $E$ for the distribution of the marginal of $U_A+U_H \in H \times (G/H)$ taking values in $G/H$. Because $U_A$ and $U_H$ are independent, we have \[(U_A+U_H)(g) = E(b(g))/|H|\] for every $g \in G$, and hence in particular that the marginals of $U_A+U_H$ taking values in $H$ and in $G/H$ are independent. This provides \[\H(U_A+U_H) = \H(U_H) + \H(E)\] by the equality case of \eqref{subadditivity of entropy} and from \eqref{first inequality} we hence obtain $\H(E) \le 2\a$. In the expression \[\H(E) = - \sum_{b \in G/H} E(b) \log E(b)\] of the entropy of $E$, every $b \in G/H$ satisfying $E(b) \le 1/2$ contributes at least $(\log 2) E(b)$. The total contribution of all such $b$ is therefore at least $T \log 2$, where $T$ is the sum of all $E(b)$ with $E(b) \le 1/2$. As there can be only at most one $b$ satisfying $E(b) > 1/2$, we obtain that there exists an element $b \in G/H$ such that \[\sum_{b’ \neq b} E(b’) \le 2\a/\log 2.\] The probability distribution $E$ of the marginal of $U_A + U_H$ taking values in $G/H$ is also the analogous distribution with $U_A+U_H$ replaced by $U_A$, so we have established that \[|A \setminus (H \times \{b\})| \le (2\a/\log 2) |A| \] and hence that \[|A \setminus (H \times \{b\})| \le (4\a/\log 2) |H| \] using that $\a \le 1/10$. Choosing $x \in G$ with $b(x) = b$ and furthermore using \eqref{two differences}, \eqref{size difference} we obtain \[|A \Delta (H + \{x\})| \le |H|(2\a+(4\a/\log 2)).\] The right-hand side is at most $10 \a |H|$, which finishes the proof. \end{proof}

All is now ready for us to deduce Theorem \ref{Main theorem}.

\begin{proof} [Proof of Theorem \ref{Main theorem}]

Letting $\g$ be as in Corollary \ref{Corollary with X=Y}, we take $\t = \min(1,100\g)$. With this value of $\t$, let $\e,G,A$ be as in Theorem \ref{Main theorem}. Applying Lemma \ref{bound on Delta H} shows that \[\Delta \H (A) \le 2\e \log (\e^{-1}|A|).\] Provided that the right-hand side is at most $\g$, Corollary \ref{Corollary with X=Y} then provides the bound \[d(U_A,U_H) \le 24\e \log (\e^{-1}|A|)\] for some finite subgroup $H$ of $G$. Provided that this newer right-hand side is at most $1/10$, Lemma \ref{bound on A Delta H} finally establishes that \[|A \Delta (H + \{x\})| \le 240 \e \log (\e^{-1}|A|) |H|\] for some $x \in G$. \end{proof}


\begin{thebibliography}{9}

\bibitem{Breuillard Green Tao}

E. Breuillard, B. Green, T. Tao, \textit{The structure of approximate groups}, Publ. Math. IHES \textbf{116} (2012), 115-221.

\bibitem{GGMT2}

W. T. Gowers, B. Green, F. Manners, T. Tao, \textit{Marton's Conjecture in abelian groups with bounded torsion}, 	arXiv:2404.02244.

\bibitem{GGMT1}

W. T. Gowers, B. Green, F. Manners, T. Tao, \textit{On a conjecture of Marton}, arXiv:2311.05762.

\bibitem{GMT}

B. Green, F. Manners, T. Tao, \textit{Sumsets and entropy revisited}, arXiv:2306.13403.

\bibitem{Ruzsa}

I. Ruzsa, \textit{Sumsets and entropy}, Random Struct. Alg., \textbf{34} (2009), 1-10.

\bibitem{Sanders F_2^n}

T. Sanders, \textit{A note on Freiman's theorem in vector spaces}, Combin. Probab. Comput. \textbf{17} (2008), 297-305.

\bibitem{Sanders}

T. Sanders, \textit{Approximate (Abelian) groups}, European Congress of Mathematics, Eur. Math. Soc., Zürich (2013), 675-689. 

\bibitem{Tao}

T. Tao, \textit{Sumset and inverse sumset theory for Shannon entropy}, Combin. Probab. Comput. \textbf{19} (2010), no. 4, 603–639.









\end{thebibliography}
\end{document}